\documentclass[12pt]{article}
\newcommand{\version}{\today}

\usepackage{amsthm,amsfonts,amsmath, amscd}

\swapnumbers
                              
\theoremstyle{plain}
\newtheorem{thm}{THEOREM}[section]

\newtheorem{cl}[thm]{COROLLARY}

\theoremstyle{definition}
\newtheorem{defi}[thm]{DEFINITION}

\theoremstyle{remark}
\newcommand{\upchi}{\raise1pt\hbox{$\chi$}}
\newcommand{\R}{{\mathord{\mathbb R}}}

\newcommand{\hn}{{\mathord{\widehat{n}}}}

\newcommand{\tr}{{\rm Tr}}
\renewcommand{\|}{{\Vert}}
\numberwithin{equation}{section}
\newcommand{\un}{{\rm 1\kern -2.5pt l}}

\begin{document}
\markboth{\scriptsize{CL \version}}{\scriptsize{CL February 20, 2015}}
\def\mn{{\bf M}_n}
\def\hn{{\bf H}_n}
\def\hnp{{\bf H}_n^+}
\def\hmnp{{\bf H}_{mn}^+}
\def\H{{\mathcal H}}
\title{{\sc A remainder term for H\"older's inequality for matrices and quantum entropy inequalities}}

\author{\vspace{5pt} Eric A. Carlen$^1$\\
\vspace{5pt}\small{$1.$ Department of Mathematics, Hill Center,}\\[-6pt]
\small{Rutgers University,
110 Frelinghuysen Road
Piscataway NJ 08854-8019 USA}\\
}
\date{\version}
\maketitle 
\footnotetext                                                                         
[1]{Work partially
supported by U.S. National Science Foundation
grant DMS 1501007.   }

\begin{abstract}
We prove a sharp remainder term for H\"older's inequality for traces as a consequence of the uniform convexity properties of the 
Schatten trace norms. We then show how this implies 
a novel family of Pinsker type bounds for the quantum Renyi entropy. Finally, we show how the sharp form of
the usual quantum Pinsker inequality for relative entropy may be obtained as a fairly direct consequence of uniform convexity.

\end{abstract}

\medskip
\leftline{\footnotesize{\qquad Mathematics subject
classification numbers: 26B25, 94A17}}
\leftline{\footnotesize{\qquad Key Words: density matrix, entropy, uniform convexity}}

\medskip


\maketitle

\section{Introduction}

For any $n \times n$ matrix $A$,  define
$|A| = (A^*A)^{1/2}$,
and for $1 \leq p < \infty$,
${\displaystyle \|A\|_p = \left(\tr |A|^p\right)^{1/p}}$.
If $\sigma_1 \geq \dots \geq \sigma_n$ are the singular values of $A$, then 
$\|A\|_p = \left(\sum_{j=1}^n \sigma_j^p\right)^{1/p}$.
For $p=\infty$, $\|A\|_\infty$ is simply the operator norm of $A$, which is also the largest singular value of $A$. 
It is well-known that $\|\cdot \|_p$ is a norm, the Schatten $p$ norm,  on $M_n$, the space of $n\times n$ matrices. 
The space $C_p$ is the space $M_n$  of $n\times n$ complex matrices equipped with this norm. 

The Schatten  norms are in many ways close analogs of the $\ell_p$ norms. In particular, one has the analog of H\"older's inequality
$$\left | \tr [AB]\right| \leq \|A\|_p\|B\|_{p'}$$
where $1/p+1/p' =1$. Whenever $p$ and $p'$ appear together below, it is assumed that $1/p+1/p' =1$.
For all $1\leq p \leq 1$, a simple argument using the singular value decomposition shows that 
\begin{equation}\label{dual}
\|A\|_p = \sup\left\{ \Re\left(\tr[AB]\right)\ :\ \|B\|_{p'} = 1 \right\}\ ,
\end{equation}
and in fact, the supremum is achieved. 
The Minkowski inequality; i.e., the fact that our norms are norms, follows in the usual way. For $1< p < \infty$,
and non-zero $A \in C_p$, define
\begin{equation}\label{dual1}
\mathcal{D}_p(A) = \|A\|_p^{1-p}|A|^{p-1}U^*
\end{equation}
where $A = U|A|$ is the polar decomposition of $A$.  Then one readily checks that
\begin{equation}\label{dual2}
\|\mathcal{D}_p(A)\|_{p'} =1 \qquad{\rm and}\qquad \tr[\mathcal{D}_p(A)A] = \|A\|_p\ .
\end{equation}
Thus, for $1< p < \infty$, the supremum in (\ref{dual}) is a maximum, and the maximum is attained at $\mathcal{D}_p(A)$.  It then follows from (\ref{dual}) that for all $A,B\in M_n$ and all $t\in \R$,
$$\|A+t B\|_p \leq \tr[\mathcal{D}_p(A)(A+tB)]  =  \|A\|_p + t\,\tr[\mathcal{D}_p(A)] B \ .$$
Likewise, writing $A = (A+tB) - tB$,
$$\|A\|_p \leq \tr[\mathcal{D}_p(A+tB)(A)]  =  \|A +tB\|_p - t\,\tr[\mathcal{D}_p(A+tB)] B \ .$$
Thus, provided that $t\|\mathcal{D}_p(A+tB) -\mathcal{D}_p(A)\|_p = o(|t|)$,  we have that 
$$\big|\|A+tB\|_p  - \|A\|_p - t\,\tr[\mathcal{D}_p(A) B]\big| = o(|t|)\ ,$$
and this says that the norm function $A \mapsto \|A\|_p$ is (Fr\'echet) differentiable for $1< p < \infty$, and
that $\mathcal{D}_p(A)$ is the derivative at $A\in M_n$.  In fact, for $1< p < \infty$, the map $A \mapsto 
\mathcal{D}_p(A)$ is H\"older continuous, and the modulus of continuity has been given in \cite{CL}. Thus, $A\mapsto 
\mathcal{D}_p(A)$ is the gradient of the norm function $A\mapsto \|A\|_p$ for $1 < p < \infty$, and this is the reason for the notation
using $\mathcal{D}$. 

The map $A \mapsto \mathcal{D}_p(A)$ is closely related to the {\em non-commutative Mazur map} studied in \cite{AP} an \cite{ER}.
For $1\leq p,q\leq \infty$, the {\em Mazur map} $\mathcal{M}_{p,q}$ is defined on $M_n$ by $\mathcal{M}_{p,q}(A) = A|A|^{(p-q)/q}$. 
For $q=p'$, $(p-q)/q = p-2$, and hence
$$\mathcal{M}_{p,p'}(A) = \|A\|_p^{p-1}(\mathcal{D}_p(A))^*\ .$$
Sharp H\"older continuity bounds on $\mathcal{M}_{p,q}$ in a very general von Neumann algebra setting are proved in \cite{ER}, which can be consulted for further references. 

The norm gradient maps, which are the Mazur maps for $q=p'$, normalized to be homogeneous of degree one, are the focus of this note which concerns another setting in which they arise. Our first result is a quantitative remainder term for the tracial H\"older inequality. From this we shall deduce several quantum entropy inequalities.  

The next theorem is a non-commutative analog of a theorem proved in \cite{CFL} in the commutative context of $L_p$ spaces for Lebesgue integration. The proof simply uses the sharp uniform convexity properties of the $C_p$ norms proved in \cite{BCL} in place of the corresponding sharp uniform convexity properties of the $C_p$ norms that were used in \cite{CFL}.

\begin{thm}[H\"older's inequality with remainder]\label{holder}
Let $1 < p\leq 2$. Let $A$ be a unit vector in $C_p$, and let $B$ be a unit vector in $C_{p'}$. 
Let $\theta \in [0,2\pi)$ be chosen such that  $e^{i\theta}\tr[AB]$ is non-negative. Then 
we have both
\begin{equation}\label{main1}
\left|\tr[AB]\right| \leq 1 -  \frac{p-1}{4}\|{\mathcal D}_{p'}(B) -e^{i\theta}A\|_{p}^2\ ,
\end{equation}
and  
\begin{equation}\label{main2}
\left|\tr[AB]\right| \leq 1 -  \frac{1}{p'\ 2^{p'-1}}\| e^{i\theta}B - {\mathcal D}_{p}(A)\|_{p'}^{p'}\ .
\end{equation}
The exponents $2$ and $p'$ on the right sides of \eqref{main1} and \eqref{main2} are best possible. 
\end{thm}

\begin{proof} By (\ref{dual1}) and the choice of $\theta$,
${\displaystyle 1+  e^{i\theta}\tr[AB] = \tr [ (\mathcal{D}_{p'}(B) + e^{i\theta}A)B]}$.
Therefore, by H\"older's inequality and  the choice of $\theta$,
\begin{equation}\label{main10}
1 + \left|\tr[AB]\right|  \leq  \|{\mathcal D}_{p'}(B) +e^{i\theta}A\|_{p} \leq 2 \left\Vert \frac{{\mathcal D}_{p'}(B) +e^{i\theta}A}{2}\right\Vert_p \ .
\end{equation}
Now apply the optimal $2$-uniform convexity inequality \cite{BCL}, valid for $1 < p \leq 2$, and unit vectors $X,Y\in C_p$:
\begin{equation}\label{2u}
\left\Vert\frac{X+Y}{2}\right\Vert_p \leq 1 - \frac{p-1}{2}\left\Vert\frac{X-Y}{2}\right\Vert^2_p\ .
\end{equation}
This leads directly to  (\ref{main1}).  The proof of (\ref{main2}) is similar except that one uses
$$\left\Vert\frac{X+Y}{2}\right\Vert_p \leq 1 - \frac{1}{p}\left\Vert\frac{X-Y}{2}\right\Vert^p_p\ .$$
valid for $ 2 \leq p$, and unit vectors $X,Y\in C_p$ \cite{BCL}.  The fact that the exponents are the best possible follows from the fact that this is true in the commutative case, and the proof of this may be found in Theorem 3.1 of \cite{CFL}.
\end{proof}

\section{Application to entropy}

Recall that for $\alpha\in (0,1)$, the Renyi $\alpha$-relative entropy for $\rho$ with respect to $\sigma$
is the quantity
\begin{equation}\label{rendef1}
D_\alpha(\rho||\sigma) = \frac{1}{\alpha -1}\log\left(\tr[\rho^\alpha \sigma^{1-\alpha}]\right)\ .
\end{equation}
Recall also that 
$$\lim_{\alpha\to 1} D_\alpha(\rho||\sigma)   = D(\rho||\sigma)  := \tr[\rho (\log \rho - \log \sigma)]\ ,$$
the von Neumann relative entropy. Pinsker's inequality for the von Neumann relative entropy states that
\begin{equation}\label{pinsk}
D(\rho||\sigma)  \geq \frac12 \|\rho-\sigma\|_1^2\ .
\end{equation}

We now show that  Theorem~\ref{holder} gives a Pinsker type inequality for the Renyi entropy from which
(\ref{pinsk}) can be derived in the limit $\alpha \to 1$. By the definition (\ref{rendef1}), for $\alpha\in (0,1)$, 
entails an upper bound on $\tr[\rho^\alpha \sigma^{1-\alpha}]$ implies a lower bound on
$D_\alpha(\rho||\sigma)$.

\begin{thm}
Let $\rho$ and $\sigma$ be density matrices in $M_n$ for some $n$, and let $1 < p \leq 2$.
\begin{equation}\label{main40}
\tr[ \sigma^{1- 1/p} \rho^{1/p} ] \leq 1 - \frac{p-1}{4} \|\rho^{1/p} - \sigma^{1/p}\|_p^2\ ,
\end{equation}
and
\begin{equation}\label{main40B}
\tr[ \sigma^{1- 1/p} \rho^{1/p} ] \leq 1 - \frac{1}{p'\ 2^{p'-1}} \|\rho^{1/p'} - \sigma^{1/p'}\|_{p'}^{p'}\ .
\end{equation}
\end{thm}

\begin{proof} Define
$A = \rho^{1/p}$ and $B = \sigma^{1/p'}$ so that $A$ and $B$ 
are unit vectors in $C_p$ and $C_{p'}$ respectively.  First note  that
${\mathcal D}_{p'}(B) = B^{1/(p-1)} = \sigma^{1/p}$. Hence (\ref{main40}) follows directly from (\ref{main1}).
Next,  note that ${\mathcal D}_p(A) = A^{p-1} = \rho^{1-1/p}$. 
Hence (\ref{main40B}) follows directly from (\ref{main2}).
\end{proof}

\begin{cl}\label{strpinren} For all $\alpha\in [1/2,1)$, 
\begin{equation}\label{main51}
D_\alpha(\rho||\sigma) \geq \frac{1}{4\alpha} \|\rho^\alpha - \sigma^\alpha\|_{1/\alpha}^2\ .
\end{equation}
\end{cl}

\begin{proof}
Take $p := 1/\alpha$ and $\alpha\in [1/2,1)$ so that $p\in (1,2]$. Then (\ref{main40}) yields 
(\ref{main51}).   
\end{proof}

We could of course use (\ref{main2}) to treat the cases $\alpha\in (0,1/2)$ in an analogous way; the result would be similar, but the exponent on the right would be $1/\alpha$ in place of $2$. 

Lower bounds on $D_\alpha(\rho||\sigma) $ in terms of $\|\rho-\sigma\|_1$ are known in the classical case, and easily generalize to the
quantum case, but these bounds are weaker than the bounds provided by Corollary~\ref{strpinren}.
It is known \cite{LG}  in the classical case
($\rho$ and $\sigma$ commuting) that
\begin{equation}\label{gr}
D_\alpha(\rho||\sigma) \geq \frac{\alpha}{2}\|\rho-\sigma\|_1^2\ .
\end{equation}
As a quite direct consequence of the Lieb Concavity Theorem \cite{wy} which says that
$(\rho,\sigma)\mapsto \tr[\rho^\alpha\sigma^{1-\alpha}]$ is concave for $\alpha\in [0,1]$, this is also valid in the
quantum case.  
Indeed, let $P$ be the projector onto
the range of $(\rho -\sigma)_+$. Let $U$ be any unitary that commutates with $P$.
Then
$$\tr[(U\rho U^*)^\alpha (U \sigma U^*)^{1-\alpha}] = \tr[\rho^\alpha\sigma^{1-\alpha}]\ .$$
By a theorem of Uhlmann \cite{U}, there is a finite set of such unitaries such that if we define 
$\widehat{\rho}$ and $\widehat{\sigma}$ be the averages of $U\rho U^*$ and 
$U \sigma U^*$ respectively over all  unitaries $U$ in our set, then  $\widehat{\rho}$ and $\widehat{\sigma}$
both belong to the algebra generated by $P$, and hence for some $p,q\in [0,1]$,
$$\widehat{\rho} = \frac{p}{\tr[P]}P +  \frac{(1-p)}{\tr[I-P]}(I-P) \qquad{\rm and}\qquad 
\widehat{\sigma} = \frac{q}{\tr[P]}P +  \frac{(1-q)}{\tr[I-P]}I-P\ .$$
Then the Lieb Concavity Theorem implies that
$\tr[\widehat{\rho}^\alpha \widehat{\sigma}^{1-\alpha}] \geq \tr[\rho^\alpha\sigma^{1-\alpha}]$.
Hence
$D_\alpha(\rho||\sigma) \geq D_\alpha(\widehat{\rho}||\widehat{\sigma})$.
However, since $\widehat{\rho}$ and $\widehat{\sigma}$ commute, the classical bound (\ref{gr}) applies to yield 
${\displaystyle D_\alpha(\widehat{\rho}||\widehat{\sigma}) \geq \frac{\alpha}{2}\|\widehat{\rho}-
\widehat{\sigma}\|_1^2}$, 
and one easily sees that $\|\rho-\sigma\|_1 = \|\widehat{\rho} -\widehat{\sigma}\|_1$. 
Hence, (\ref{gr}) is valid in the quantum setting as well.

We now show that (\ref{main51}) improves upon (\ref{gr}): The ratio of  
$\|\rho^{\alpha} - \sigma^{\alpha}\|_{1/\alpha}$ to 
$\|\rho - \sigma\|_1$ can be arbitrarily large, and that the 
ratio of $\|\rho - \sigma\|_1$  to 
  $\|\rho^{\alpha} - \sigma^{\alpha}\|_{1/\alpha}$ is bounded above by a finite constant. 
  
For the first of these points, an example suffices. Let $\epsilon \in (0,1/2)$, and define
$$\rho = \left[\begin{array}{cc} \tfrac12 & 0\\ 0 & \tfrac12\end{array}\right] 
\qquad{\rm and}\qquad 
\sigma =  \left[\begin{array}{cc} \tfrac12+\epsilon & 0\\ 0 & \tfrac12-\epsilon \end{array}\right] \ .$$
Then $\|\rho - \sigma\|_1 =2 \epsilon$, while 
$\|\rho^{\alpha} - \sigma^{\alpha}\|_{1/\alpha} = \left(2\alpha\epsilon\right)^{\alpha} + \mathcal{O}(\epsilon^{2\alpha})$. It follows that
\begin{equation}\label{comp1}
\frac{\phantom{1}\|\rho^{\alpha} - \sigma^{\alpha}\|_{1/\alpha}}{\|\rho - \sigma\|_1} =  \frac{\alpha^\alpha}{2\epsilon^{1-\alpha}} +
\mathcal{O}(\epsilon^{2\alpha-1})
\end{equation}
This shows that (\ref{main51}) can provide a much stronger bound than (\ref{gr}). 

Lemma 2.3 in \cite{ER} says (in particular) that for all positive $A,B\in M_n$, and all $\alpha \in (0,1)$,
\begin{equation}\label{ricard1}
\frac{\alpha}{3}\|A -B\|_1 \leq \|A^\alpha - B^\alpha\|_{1/\alpha} \max\{ \|A^\alpha\|_{1/\alpha}\ ,  
 \|B^\alpha\|_{1/\alpha} \}\ .
\end{equation}
Applying this with $A = \rho$ and $B = \sigma$, $(\alpha/3)\|\rho - \sigma\|_1 
\leq \|\rho^\alpha - \sigma^\alpha\|_{1/\alpha}$. Combining this with (\ref{main51}) 
yields, for $\alpha\in [1/2,1)$, 
\begin{equation}\label{main515}
D_\alpha(\rho||\sigma) \geq \frac{\alpha}{36} \|\rho - \sigma\|_1^2\ ,
\end{equation}
which is (\ref{gr}) apart from a constant that is worse by a factor of $18$. Thus apart from the constant, (\ref{main51})
implies (\ref{gr}). Part of the discrepancy in the constants is due to the constant in (\ref{ricard1}), but part also is due to the fact that we have not yet made optimal use of the uniform convexity bounds. 

As we now show, more can be gleaned from the argument
that we used to deduce a remainder term for H\"older's inequality from uniform convexity bounds. We now prove a variant of
Theorem~\ref{holder} and show that using this variant, we may obtain the full non-commutative Pinsker inequality; i.e.,
the $\alpha \to 1$ limit of (\ref{gr}), with the exact constants. This derivation shows that the sharp form of 
Pinsker's inequality is actually a fairly direct consequence of the uniform convexity properties of the $C_p$ spaces.

\section{Pinsker's inequality and uniform convexity}

Since
$$\lim_{\alpha\uparrow 1}D_\alpha(\rho||\sigma) = D(\rho||\sigma) = \tr[\rho(\log \rho - \log \sigma)]\ ,$$
taking the limit $\alpha\uparrow 1$ in  (\ref{strpinren}) yields
${\displaystyle D(\rho || \sigma) \geq \frac{1}{4} 
 \|\rho - \sigma\|_1^2}$. 
 This  is Pinsker's inequality \cite{Pin,R}, except that it is not in the sharp form which has a factor of $1/2$
 in place of the $1/4$ on the right, which is what one obtains from (\ref{gr}) in the limit $\alpha\uparrow 1$. 
 However, one can   recover the sharp form of Pinsker's inequality
 from the optimal $2$-uniform convexity inequality by going back to the proof 
 of  Theorem~\ref{holder} and noting that we gave something up arriving at (\ref{main10}) 
 by applying the usual H\"older inequality without taking the remainder into account. 
 
 \begin{defi} Let $\mathcal{P}$ be the set of functions $A(p)$ from $[1,2]$ into the positive $n\times n$ matrices 
 such that $\lim_{p\to 1}A(p) = A(1)$ in $C_1$ and such that $\|A(p)\|_p =1$ for each $p\in [1,2]$. 
 \end{defi}
 
 For example, let $\rho$ be any density matrix in $M_n$. Then $A(p) := \rho^{1/p} \in \mathcal{P}$. Moreover,
 if $A(p)$ and $B(p)$ belong to  $\mathcal{P}$, then so does $(A(p) + B(p))/\|A(p)+ B(p)\|_p$.

 \begin{thm}[Variant of H\"older's inequality with remainder]\label{holder2}
Let $1 < p\leq 2$.  Then for all $A(p)$ and $B(p)$ in  $\mathcal{P}$, and any constant $K < 1/2$,
\begin{equation}\label{main21}
\tr[A(p)B^{p-1}(p)] \leq 1 - K(p-1)\|A(p) - B(p)\|_p^2 + o(p-1)\ .
\end{equation}
\end{thm}

\begin{cl}[Pinsker's Inequality for Density Matrices]
For all density matrices $\rho$ and $\sigma$ in $M_n$, 
$$D(\rho||\sigma) \geq \frac12 \|\rho-\sigma\|_1^2\ .$$
\end{cl}

\begin{proof}  Take $A(p) - \rho^{1/p}$ and $B(p) = \sigma^{1/p}$.  By Theorem~\ref{holder2}, for all $K < 1/2$, 
$$\tr[ \sigma^{1- 1/p} \rho^{1/p} ] \leq 1 - K(p-1(\|\rho^{1/p} - \sigma^{1/p}\|_p^2 + o(p-1)\ .$$
Rearranging terms as above, and taking $p\to 1$, we obtain
$$D(\rho||\sigma) \geq K \|\rho-\sigma\|_1^2\ .$$
since $K < 1/2$ is can be arbitrarily close to $1/2$, the inequality is proved. 
\end{proof} 

\begin{proof}[Proof of Theorem~\ref{holder2}]  We note that for $C= 1/4$, 
(\ref{main21}) is valid by Theorem~\ref{holder}.  Next, supposing that  (\ref{main21}) is valid for some constant 
$K$, we show that it is also valid when $K$ is replaced by $(K+1/2)/2$. Iterating this yields the claimed result. 

Therefore, let us make the inductive assumption that  (\ref{main21}) is valid for some constant 
$K$.  We have
\begin{equation}\label{main31}
 1 + \tr \left[  B^{1-p}(p) A(p) \right] = \tr[B^{1-p}(p)(A(p) + B(p))] =  \tr[B^{1-p}(p)C(p))]\|A(p) + B(p)\|_p
 \end{equation}
where
$$C(p) = \|A(p) + B(p)\|_p^{-1}(A(p)+B(p)) \in \mathcal{P}\ .$$
By hypothesis,
$$\tr[B^{1-p}(p)C(p))]  \leq 1 - K(p-1)\|B(p) - C(p)\|_p^2 + o(p-1)\ ,$$
and since $\lim_{p\downarrow 1}\|A(p) + B(p)\|_p =2$, 
$$\|B(p) - C(p)\|_p = \left\Vert B(p) - \frac{A(p)+B(p)}{2}\right\Vert_p + o(1) = \frac12 \|A(p) - B(p)\|_p + o(1)\ .$$
Combining this with the previous bound,
\begin{equation}\label{main30}
\tr[B^{1-p}(p)C(p))]  \leq 1 - (p-1)\frac{K}{4} \|A(p) - B(p)\|_p^2 + o(1)\ .
\end{equation}
By the $2$-uniform convexity inequality, 
$$
\|A(p) + B(p)\|_p \leq 2- \frac{p-1}{4}\|A(p) - B(p)\|_p^2\ .$$
Using this and (\ref{main30}) in (\ref{main31}), we obtain
\begin{eqnarray}
 1 + \tr \left[  B^{1-p}(p) A(p) \right]  &\leq&  2 \left( 1 - \frac{p-1}{8}\|A(p) - B(p)\|_p^2\right)\left(
1 - (p-1)\frac{K}{4} \|A(p) - B(p)\|_p^2 + o(1)\right)\nonumber\\
&\leq& 2 - (p-1)\left(\frac14 + \frac{K}{2}\right) \|A(p) - B(p)\|_p^2 + o(p-1)\ .\nonumber
\end{eqnarray}
Thus, in (\ref{main21}), we may replace $K$ by $(K + 1/2)/2$, and the validity is maintained. 
\end{proof}

\noindent{\bf Acknowledgement} This results in this paper were obtained while the author was visiting at the I.M.A. in 
Minnesota during Spring 2015.

 \end{document}